\newcommand{\columnsversion}[2]{#2}{}  

\columnsversion{
\documentclass{acm_proc_article-sp}
}{
\documentclass[11pt,letterpaper]{article}

\topmargin=-0.35in \topskip=0pt \headsep=15pt \oddsidemargin=0pt
\textheight=9in \textwidth=6.52in \voffset=0in
}

\usepackage{amsmath,amsfonts, epsfig}
\usepackage{amssymb,amstext,xspace}
\usepackage{bbm, color}
\usepackage{algorithm,algorithmic}

\usepackage{color}              
\usepackage{epsfig}
\usepackage{url}

\usepackage{color-edits}
\addauthor{rpl}{red}	
\addauthor{mb}{blue}
\addauthor{nn}{green}

\newtheorem{theorem}{Theorem}
\newtheorem{lemma}[theorem]{Lemma}

\newtheorem{example}[theorem]{Example}

\newtheorem{corollary}[theorem]{Corollary}

\def\qedsymbol{$\square$}
\newcommand{\bg}[1]{\medskip\noindent{\textsc{#1}}}
\newcommand{\ed}{{\hfill\qedsymbol}\medskip}
\newenvironment{proofof}[1]{\bg{Proof of #1 : }}{\ed}

\newenvironment{proof}{\bg{Proof : }}{\ed}

\columnsversion{}{
\newcommand{\email}[1]{\texttt{#1}}
}

\newcommand{\R}{\ensuremath{\mathbb R}}
\newcommand{\Z}{\ensuremath{\mathbb Z}}
\newcommand{\comment}[1]{}
\newcommand{\E}{\mathbb{E}}
\newcommand{\abs}[1]{\vert{#1}\vert}
\newcommand{\Prob}{\mathbb{P}}

\newcommand{\argmax}{\operatorname{argmax}}
\newcommand{\nash}{\textsc{Nash}}
\newcommand{\mnash}{\textsc{mNash}}
\newcommand{\cnash}{\textsc{cNash}}
\newcommand{\pareto}{\textsc{Par}}

\newcommand{\one}{\ensuremath{\mathbb{I}}}

\begin{document}

\columnsversion{}{
\setcounter{page}{0}
}

\title{Price Competition in Online Combinatorial Markets}
%
%
%
%
%

\columnsversion{
\numberofauthors{3}
}{}
\author{
\columnsversion{\alignauthor}{} Moshe Babaioff\\	
       \columnsversion{\affaddr}{}{Microsoft Research}\\
       \email{moshe@microsoft.com}
\columnsversion{\alignauthor}{\and} Noam Nisan\\
       \columnsversion{\affaddr}{}{Microsoft Research and HUJI}\\
       \email{noamn@microsoft.com}
\columnsversion{\alignauthor}{\and} Renato Paes
Leme\columnsversion{\titlenote}{\footnote}{This work was
done while the author was a post-doctoral researcher at Microsoft Research
Silicon Valley.}\\
       \columnsversion{\affaddr}{}{Google Research NYC}\\
       \email{renatoppl@google.com}
}
\columnsversion{}{
\date{}
}

\maketitle

\begin{abstract}
We consider a single buyer with a combinatorial preference that would like to
purchase related products and services from different vendors,
where each vendor supplies exactly one product.
We study the general case where subsets of products can be substitutes as well as complementary and
analyze the game that is induced on the vendors, where a vendor's strategy is the price that he asks
for his product.
This model generalizes both Bertrand competition (where vendors are perfect substitutes) and Nash
bargaining (where they are perfect complements), and captures a wide variety of scenarios that can appear
in complex crowd sourcing or in automatic pricing of related products.

We study the equilibria of such games and show that a pure efficient equilibrium always exists.
In the case of submodular buyer preferences we fully characterize the set of pure Nash equilibria,
essentially showing uniqueness. For the even more restricted
``substitutes'' buyer preferences we also prove uniqueness over {\em mixed} equilibria.
Finally we begin the exploration of natural generalizations of our setting such
as when services have costs, when there are multiple buyers or uncertainty about
the the buyer's valuation, and when a single vendor supplies multiple products.

\end{abstract}

\columnsversion{

}{
\newpage}

\section{Introduction}


It is a common practice in electronic commerce for sellers to use algorithmic
techniques to automatically adjust the price of goods and services in response
to prices posted by other sellers. Such practices are specially in evidence
when malfunctioning algorithms lead to rather amusing results: Eisen
\cite{eisen_blog}
observed that the biology book 'The Making of a Fly' was priced at 23.6 million
dollars by a seller in Amazon and at 18.6 million dollars by another seller.
By observing how prices changed over time, Eisen concluded that this result was
reached by sellers recurrently applying price update rules $p_1 \leftarrow .998
\cdot p_2$ and $p_2 \leftarrow 1.27 \cdot p_1$, without taking the valuation of the prospective buyers into account.
Prices reached these high values as after each time both sellers updated their prices,
the prices grew by a factor of $.998 \cdot 1.27 > 1$.

Such absurd outcomes are uncommon (but this example is by no means unique) and are often the result of
faulty algorithms, but they hint at at a complex game between different sellers
that adjust their prices based on the price of their competitors and based on
the dependencies among products (substitutabilities and complementarities).
The fact that sellers frequently update their prices in response to other
sellers and the resulting fluctuations have been traditionally observed in
financial markets as well as in prices for airline tickets. A recent news
article \cite{wsj_article} observed that the same technique is being adopted by
online retailers for more mundane consumer goods. The article tracks the price
of a microwave across three different online retailers: Amazon, Best Buy and
Sears. They observed the price being changed $9$ times during the day by Amazon
and twice by Best Buy, in response to Amazon prices.

Our goal in this paper is to formally model the \emph{pricing game}
played by the sellers and study properties of its equilibria as a function of
the dependencies among goods/services offered by the sellers.
Such phenomena are not an exclusive feature of internet markets. Traditional
economic models such as Bertrand competition \cite{bertrand} and Nash Bargaining
\cite{nash_bargaining}
explore versions of this question. Two aspects of online markets, however, are
particularly relevant. The first is operational: the internet provides
quick access to the information on prices of other sellers and allows sellers (or software agents acting on their behalf) to
respond in real-time. The second aspect is structural: products and services
exhibit complex dependencies. We exemplify below some situations in which such
dependencies arise:

\begin{itemize}
 \item \emph{cloud services}: building a cloud service involves assembling
various components such as storage, databases, bandwidth, etc. . Each of those
is available through a variety of vendors. Complex dependencies arise from the
fact that some product are more compatible with each other and some are less.
For example, a certain database might place some requirements on storage and
bandwidth, making some combinations of services infeasible.
 \item \emph{information and data streams}: it is common for internet
advertisers to use behavioral and demographic information about the users to
bid on ads and to choose which content to display. This information is
typically purchased from third party data providing agencies who collect and
curate databases for this purpose. Each data providing agency has a partial
view of the user and complex queries often require advertisers to purchase data
from multiple vendors. Two pieces of data might exhibit subtitutability (e.g.
if they provide the same information about the user) or might be complements
(e.g. if they have a common attribute that allows the advertiser to perform a
\emph{join} and link two different views about the same user).
 \item \emph{crowdsourcing and online marketplaces}: online platforms such as
oDesk allow workers to post hourly wages and buyers to assemble teams of
workers with different skills. Workers with similar skills are substitutes
while two workers that are in same geographic location and can more easily work
together might be complements.
 \item \emph{routing}: given a source and a destination in a network, one needs
to buy a path connecting them. Two edges needed to complete a path are
complements. Two parallel edges are perfect substitutes. This corresponds to
the path auctions of Immorlica et al \cite{immorlica05}.
\end{itemize}

{\bf Our results.}
We first define a simple model of a \emph{pricing game}, in which we wish to
capture the essential aspects of the competition among sellers. Our basic game
consists of $n$ sellers and one buyer: the buyer has a public\footnote{the case in which there is uncertainty about the buyer's valuation is briefly discussed in Section~\ref{subsec:value_uncertainty}.} combinatorial valuation
over the items and reacts to item prices by purchasing the bundle of
items/services that maximizes his utility (total value minus total price),
breaking ties according to a fixed rule. The sellers are strategic agents and
each seller's strategy consists of setting a price for the service he is providing.

Our first result is the existence of pure Nash equilibria for every
combinatorial valuation if the buyer breaks ties maximally, i.e., favoring
supersets. For arbitrary tie breaking rules, we show existence of
$\epsilon$-Nash equilibria for all $\epsilon > 0$. {Both results rely on a structural characterization of the set of equilibria.  }

Although pure Nash equilibrium always exists, it might fail to be unique, as in the case of Nash Bargaining.
{When two sellers are perfect substitutes from the buyer's perspective, the classical Bertrand competition model asserts that the unique equilibrium consists of both buyers pricing at zero. For combinatorial valuations, are there properties encoding that items are \emph{substitutes} that lead to uniqueness of equilibrium ? We explore four common notions of substitutes in economics and study their effect on the set of equilibria of the pricing game. In increasing generality, we study the following models of substitutability: gross substitutability (a traditional notion of substitutes in economics due to Kelso and Crawford \cite{KelsoCrawford}), submodularity, XOS and subadditivity. We refer to \cite{LehmannLehmannNisan} for an extensive discussion on the various notions of substitutability. }

{First we show that for subadditive and XOS valuations, utilities of players might not be unique across all pure Nash equilibria.} If the buyer valuation is submodular, however, each player
has the same utility at every pure Nash equilibrium.\footnote{The equilibrium might fail to be unique, but only because sellers with zero utility have different strategies resulting with that utility.}
For the subclass of gross substitutes valuations a stronger claim is true:
utilities are the same across all mixed Nash equilibrium, under mild assumptions on the tie breaking rules
used by the buyer.

{\bf Extensions.} We extend our basic model in three different directions, relaxing assumptions made on the basic model:
\begin{itemize}
\item \emph{service costs}: we consider each seller as having a cost to perform the service if selected by the buyer. The utility of the seller then becomes, zero if not selected and the price posted minus his cost if selected. We show that the existence of equilibrium still holds in this setting, moreover, there is always an equilibrium maximizing the total welfare of the system, which is defined by the total value obtained by the buyer minus the total cost by the sellers. In other words, the Price of Stability\footnote{the Price of Stability measures the ratio between the optimal welfare achievable and the optimal welfare in equilibrium. The Price of Anarchy measures the ratio between the optimal welfare achievable and the worse welfare in equilibrium.} is $1$. We show that if valuations are gross substitutes, all equilibria are welfare maximizing, i.e., the Price of Anarchy is $1$. For submodular functions, however, we show that inefficient equilibria might arise. Moreover, we show that if there are costs,
 even if the valuation is submodular, the utilities in equilibria might not be unique.

\item \emph{value uncertainty and multiple buyers}: we relax the assumption that the valuation of the buyer is public knowledge of the sellers, and assume instead that sellers only know that the valuation is drawn from a certain distribution. This framework also captures the model where there are multiple buyers and sellers are not allowed to price-discriminate between them. We show that pure Nash equilibrium might not exist in this setting. Moreover, not even $\epsilon$-Nash are guaranteed to exist for small $\epsilon > 0$.

\item \emph{multiple services per seller}: in the basic model, we assumed that each service is controlled by a single seller. We consider the model where each seller sets item prices and collects utility from potentially many items. We note that even the special case of a \emph{monopolist seller}, where there is only one seller controlling all items, is quite non-trivial. We show that for this case, both the Price of Anarchy and Price of Stability are $\Theta(\log n)$. Moreover, the problem of best-responding in this game is a generalization of the Unique Coverage problem, which was shown to be $O(\log n)$-hard to approximate by Demaine et al \cite{Demaine08}. In this paper, here we provide an $O(\log n)$ approximation to the pricing problem faced by the monopolistic seller for the case where the buyer's valuation is submodular.
\end{itemize}

{\bf Conclusions and Open Problems.}
In this paper we propose a model for the price competition in combinatorial
markets generalizing traditional models such as Bertrand Competition and Nash
Bargaining to scenarios with more complex dependencies over goods. Through our
model we seek to understand which conditions are necessary so that the price
competition leads to an equilibrium and also under which conditions is this
equilibrium unique or efficient.

The lack of equilibrium where there is uncertainty about the valuation of the
buyer (Section \ref{subsec:value_uncertainty}) offers a possible explanation for
the price cycles observed in \cite{wsj_article}. Given such interpretation, it
leads to the interesting open question of trying to learn the belief of the
sellers about buyers based on their history of price updates. Our extensions
also lead to other open questions:
\begin{itemize}
\item bounding the Price of Anarchy over the set of mixed Nash equilibria when there is uncertainty about the buyer valuation.
\comment{\item bounding the Price of Anarchy over pure Nash equilibria when
there are service costs and the buyer's value is submodular.}
\item characterizing the set of Nash equilibria when there are multiple  services per seller beyond the monopolistic seller case.
\end{itemize}

{{\bf Related Work}
Our work situates in the broader agenda of understanding the structure of Nash
equilibria of games motivated by auction mechanisms -- both direct and reverse.
Closely related to our work is the study of pay-your-bid mechanisms for the 
reverse auction settings, in which a buyer solicits price quotes from different
sellers and decides to buy from a certain subset based on the quotes received.
The same problem is often phrased as hiring a team of agents, in which a firm
solicits wages from different workers and chooses a subset of them to hire. 
Immorlica et al \cite{immorlica05,immorlica10} study the Nash equilibria for
procurement settings in which a firm needs to choose a subset of workers subject
to some  feasibility constraint. Their model differs from ours in two aspects:
(i) while in our paper the buyer has a value for each subset and chooses the one
maximizing his utility, in \cite{immorlica10} the buyer is \emph{required} to buy some
feasible set. In other words, their setting corresponds to the special case of
our setting where the value of the buyer is either $-\infty$ for infeasible sets
or $\infty$ for feasible sets. (ii) their focus is on bounds on the total
payment of the buyer (which is commonly referred as \emph{frugality}), while our
focus is on properties of the equilibrium and on measuring \emph{welfare}.
Welfare makes sense in our model since there are quantitative values associated
with each outcomes rather then feasibility constraints.

{For direct auction settings, in which a seller holds a set of items and
solicits bids from various buyers on those items, Bernheim and Whinston
\cite{BernheimWhinston} study the set of equilibria of an auction where agents
bid on generic outcomes (menu auctions). Christodoulou, Kov{\'a}cs and Schapira
\cite{CKS08} study the case where outcomes are partitions over items and agents
bid on each item individually despite having combinatorial valuations over
items (item  bidding auctions). This research line was  followed-up
by a series of subsequent papers \cite{BR11,HKMN11,Feldman13,PLST12}. We refer
to \cite{dining} for a survey.}

The previously described related work focuses on equilibrium analysis for games that are likely to arise naturally in direct and reverse markets. 
A different line of research studies the \emph{mechanism design} problem, in which a (direct or reverse) auctioneer
engineers the
rules of the market in order to guarantee desired outcome such as efficiency
and frugality. For work in mechanism design more closely related to our setting,
we refer to \cite{frugal0,frugal1,frugal2,frugal3,frugal4}.}

{Another stream of related work consists of papers who extend the
traditional price competiton model due to Bertrand \cite{bertrand} to
combinatorial settings. Babaioff, Lucier and Nisan \cite{bertrand_networks}
model  the case each buyer can only access some of the sellers but not
others. Chawla and Roughgarden \cite{ChawlaR08} and Chawla and Niu
\cite{ChawlaN09} study Bertrand competition in the context of network routing. }

\section{Preliminaries}

We first define a basic version of the \emph{pricing game}, which is the
central object of study in this paper. Later we discuss and analyze extensions
of the basic model. A (basic) \emph{pricing game} is defined by a set $N = [n]$
of services (items), each controlled by a different provider (a seller) and one buyer
with valuation $v:2^{N} \rightarrow \R_+$, where $v(S)$ is the value for acquiring a subset $S\subseteq N$
of those services. We assume that $v(\emptyset) = 0$ and that the valuation is monotone
$v(S) \leq v(T)$ for $S \subseteq T$. The strategy of each seller is to set a
price $p_i$ for his service.

While facing a price vector $p$, the buyer chooses to purchase a set according
to a decision map $X : \R^n_+ \rightarrow 2^{N}$ that associated each price vector $p$
to a set that maximizes the quasi-linear utility of the buyer
$X(p) \in D(v;p) := \argmax_{S \subseteq N} v(S) - \textstyle\sum_{i \in S} p_i$. We refer to $D(v;p)$ as
the \emph{demand correspondence}. The decision map essentially fixes how ties
are broken between sets in the demand correspondence. We will say that a
decision map is \emph{maximal} if there is no $X' \in D(v;p)$ such that $X(p)
\subsetneq X'$. 

A valuation function $v$ together with a decision map define a game among the providers in which each provider strategy is to set a price $p_i$ for his service.
As in the basic model we assume that providers have no cost providing their services,\footnote{We discuss the extension to providers with costs in Section~\ref{sec:costs}.}
the utility of each provider $i$ equals to his revenue and is:
$$u_i^X(p) = p_i \cdot \one \{i \in X(p) \}$$
where $\one \{\cdot\}$ is the indicator function.

We are interested in studying the set of \emph{pure Nash equilibria} of this game:
$$\nash^X = \{p \in \R^n_+; u_i^X(p_i, p_{-i}) \geq u_i^X(p'_i, p_{-i}), \forall p'_i \in \R_+ \}$$
i.e., strategy profiles where no seller can increase his utility by deviating.
We say that a strategy profile is a (pure) $\epsilon$-\emph{Nash equilibrium} if
no seller can improve his utility by more then $\epsilon$ by deviating:
$$\nash_\epsilon^X = \{p \in \R^n_+; u_i^X(p_i, p_{-i}) \geq u_i^X(p'_i, p_{-i}) - \epsilon, \forall p'_i \in \R_+ \}$$
We also consider the sets of \emph{mixed Nash equilibria}:
$$\mnash^X = \{p \in (\Delta\R_+)^n; \E u_i^X(p_i, p_{-i}) \geq \E u_i^X(p'_i, p_{-i}), \forall p'_i \}$$
where the elements of $(\Delta\R_+)^n$ are vectors of $n$ independent random variable taking values in $\R_+$.

When the decision map $X$ is clear from the context, we omit it from the notation, for example, we refer to $u_i$, $\nash$, $\nash_\epsilon$, $\mnash$ instead of $u_i^X$, $\nash^X$, $\nash_\epsilon^X$ and $\mnash^X$.

\comment{We will also be interested in correlated Nash equilibria, in which the strategies of the sellers are correlated random variables:
$$\cnash = \left\{p \in \Delta(\R_+^n); \begin{aligned} & \E [u_i(p_i, p_{-i} ) \vert p_i] \geq \E [ u_i(p'_i, p_{-i}) \vert p_i] \\ &  \forall p_i, p'_i \in \R_+ \end{aligned} \right\}$$}

We will characterize the set of equilibria in terms of properties of the buyer's valuation $v$. Consider the following classes of valuation functions, from more general to more specific:
\begin{itemize}
\item \emph{combinatorial}: no assumptions on the valuation besides monotonicity.
\item \emph{subadditive}: $v(S\cup T) \leq v(S) + v(T), \forall S, T$.
\item \emph{XOS}: $v(S) = \max_{t \in I} \sum_{i \in S} w_{it}$ for $w_{it} \in \R_+$.
\item \emph{submodular}: $v(S\cup T) + v(S \cap T) \leq v(S) + v(T), \forall S, T$.
\item \emph{gross substitutes}: given a price vector $p \in \R^n_+$, if $S \in D(v; p)$ then for any vector $p' \geq p$, there is $T \in D(v; p')$ such that $S \cap \{j; p_j = p'_j\} \subseteq T$.
\end{itemize}

It is known that every class in the above list is a strict subclass of the previous one. We refer to Lehmann, Lehmann and Nisan \cite{LehmannLehmannNisan} for a comprehensive discussion on such classes and on their relations.

We fix some additional notation that will be useful for the rest of the paper: given a valuation $v$ and sets $S$ and $T$ such that $S \cap T = \emptyset$,  the \emph{marginal values} of $T$ with respect to $S$ is defined to be $v(T \vert S) = v(S \cup T) - v(S)$. Given a price vector $p \in \R^n_+$, denote $p(S) = \sum_{j \in S} p_j$. When clear from the context, we sometimes omit braces in the representation of sets. For example, we represent $v(\{i\} \vert S)$ by  $v(i \vert S)$, $A \cup \{i\}$ by $A \cup i$ and $S \setminus \{j\}$ by $S \setminus j$.

We keep our model as simple as possible to highlight its interesting features. Later in Section \ref{sec:extensions} we consider extensions of the basic model to incorporate service costs, multiple services provided by the same seller, a market with multiple sellers and multiple buyers and settings with incomplete information. We discuss how such additional features influence the results for the basic model.

\section{Examples, Existence and Characterization of
Equilibrium}\label{sec:existence}

First notice that the pricing game has, as special cases, the classical models
of Bertrand competition \cite{bertrand} and Nash bargaining
\cite{nash_bargaining}.
\begin{example}[Bertrand competition]
If $n > 1$, $v(\emptyset) = 0$ and $v(S) = c$ for $S \neq \emptyset$ for some
constant $c>0$, all services are perfect substitutes and the buyer has no utility
for purchasing more then one. In this case $X(p)$ will be either empty if $p_i >
c$ for all $i$ or will contain at most one service of positive price. It is known
that for this case there is a unique Nash equilibrium that corresponds to every
seller posting price $p_i = 0$.
\end{example}

\begin{example}[Nash bargaining]\label{ex:nash_bargaining}
The other extreme case is when $v(N) = c>0$ and $v(S) = 0$ for any $S
\neq N$. This models the scenario where all the services are necessary
components for the buyer. Let $X(p) = N$ if $\sum_i p_i \leq c$ and $X(p) = \emptyset$
otherwise. The set of pure Nash equilibria of this game correspond to:
$$\nash = \{p \in \R^n_+; \textstyle\sum_i p_i = c\} \cup \{p \in \R^n_+; \textstyle\sum_i p_i \geq c + \max_i p_i\}$$
In order to see that the profiles above are Nash equilibria, notice that if
$\textstyle\sum_i p_i = c$, then $X(p) = N$ and no seller has incentive to
change his price, since by decreasing his price, he can only decrease his
utility. By increasing his price, he can only make $X(p) = \emptyset$. If
$\textstyle\sum_i p_i \geq c + \max_i p_i$, then clearly $X(p) = \emptyset$ and
for any seller $i$ and price $p'_i > 0$, $X(p'_i, p_{-i})$ is still $\emptyset$,
since $\sum_{j \neq i} p_i + p'_i > \sum_{j} p_i - \max_i p_i \geq c$.

For the converse, notice that if $p \in \nash$, then either (i) $X(p) = N$ in
which case $\sum_i p_i \leq c$. If $\sum_i p_i < c$, any player can deviate to
$p'_i = c - \sum_{j \neq i} p_j > p_i$ and improve his utility. So, it must be
the case that $\sum_i p_i = c$ or: (ii) $X(p) = \emptyset$, in which case it
must be the case that no seller can decrease his price such that $p_i > 0$ and
$\sum_i p_i \leq c$, so it must be the case that $\sum_i p_i \geq c + \max_i
p_i$.
\end{example}

Now we extend the characterization of Example \ref{ex:nash_bargaining} to a
generic combinatorial valuation.

\begin{lemma}\label{lemma:characterization}
Given a price vector $p \in \R^n_+$ and $S = X(p) \in D(v;p)$, then
$p$ is a Nash equilibrium if the following two properties hold:
$$\forall i \in S, \exists T \not\ni i \text{ s.t. } v(S) - p(S) = v(T) - p(T)$$
$$\forall i \notin S, \forall T \ni i, v(S) - p(S) \geq v(T) - p(T \setminus
i)$$
Moreover, if the map $X$ is maximal, the above statement holds with \emph{'if and only if'}.
\end{lemma}

{The first condition ensures that a provider that is picked cannot gain by slightly increasing his price.
The second condition ensure that a provider that is not picked cannot gain by being picked even if he posts a positive price that is arbitrarily low.}

\begin{proof}
If a price vector $p$ is a Nash equilibrium, then for any $i \in S$ and
$\epsilon > 0$, $u_i(p_i + \epsilon, p_{-i}) \leq u_i(p_i, p_{-i})$, so it must
be the case that $u_i(p_i + \epsilon, p_{-i}) = 0$
as with $p_i + \epsilon$ seller $i$ will not sell.
So, for every $\epsilon$,
there is a set $T_\epsilon$ such that $i \notin T_\epsilon$ and $v(S) - p(S) -
\epsilon \leq v(T_\epsilon) - p(T_\epsilon)$. Taking $\epsilon_t = \frac{1}{t}$
for $t \in \Z_+$, since there are finitely many values for $T_{\epsilon_t}$,
there must be one that occurs infinitely often. Let this be $T$. Then taking the
limit as $t \rightarrow \infty$ in this subsequence one gets that: $v(S) - p(S)
\leq v(T) - p(T)$. Since $S \in D(v;p)$, then $v(S) - p(S)\geq v(T) - p(T)$ thus it must be the case that $v(S) -
p(S) = v(T) - p(T)$.

Given $i \notin S$, then it must be the case that $0 = u_i(p_i, p_{-i}) \geq
u_i(\epsilon, p_{-i})$ for every $\epsilon\geq 0$, so for any set $T \ni i$, it must be the case that $v(S)
- p(S) \geq v(T) - p(T \setminus i) + \epsilon$. Taking the limit as $\epsilon
\rightarrow 0$, we get the desired condition.

For the converse direction, \comment{in the case of a \emph{maximal} $X$ map,
assume that $p$ satisfy the properties above. Then the second
property guarantees that  $X(p) \subseteq S$, since for all sets containing an
element outside $S$ are dominated by $S$. Since the map $S \in D(v;p)$ and $X$
is maximal, then $X(p) = S$. To see that this is a Nash equilibrium,} observe
that for $i \in S$, there is no incentive to decrease his price. Also, the first
property guarantees that for any $p'_i > p_i$, $u_i(p'_i, p_i) = 0$, since there
is a set $T \ni i$ such that $v(S) - p(S) = v(T) - p(T)$. For $i \notin S$, the
second property guarantees that for any $p'_i > 0$, $X(p'_i, p_{-i}) = S$,  so
$u_i(p'_i, p_{-i}) = 0$.
\end{proof}

{One can define the \emph{welfare} $W(p)$ of a decision map $X$ at a
price vector $p$ as the sum of the utilities of the agents
(all sellers plus the buyer) when the buyer buys $X(p)$, i.e., $W(p) = v(X(p))$.
We show that for any valuation function $v$, when the decision map is maximal the set of Nash equilibrium in non-empty,
moreover, there is always a welfare maximizing  equilibrium, one in which the buyer buys all services.
The next theorem can be seen as showing that the \emph{Price of Stability} is one, i.e., it
states that there is always an equilibrium that produces maximal welfare.}

\begin{theorem}\label{cor:non_empty}
If $X$ is maximal, then the set of pure Nash equilibria with $X(p) = N$ is non-empty.
\end{theorem}

\begin{proof}
Define the set $F$ as:
$$F = \{p \in \R^n_+; p(T) \leq v(T \vert N \setminus T), \forall T \subseteq N\}$$
and notice the condition defining $F$ corresponds to $v(N) - p(N) \geq v(N
\setminus T) - p(N \setminus T)$ for all $T$, so $N \in D(v;p)$. Note that $0\in F$ and thus $F\neq \emptyset$.

Now, define the set of Pareto vectors in $F$ as:
$$\pareto(F) = \{p \in F; \not\exists p' \in F \text{ s.t. } p' \geq p \text{
and } \textstyle\sum_i p'_i > \textstyle\sum_i p_i\}$$
Since $F$ is non-empty and compact, $\pareto(F)$ is also non-empty. Now, we
argue that $\pareto(F)$ is exactly the set of Nash equilibria with $X(p) = N$.
Since the full set of sellers is selected, it is enough to argue about the
first
condition in Lemma \ref{lemma:characterization}.

If $p \in \pareto(F)$ then for any $t \in \Z_+$, $X(p_i + \frac{1}{t}, p_{-i})
\neq N$, so there must be some $T_t \subsetneq N$ and $i \notin T_t$ such that
$v(T_t) - p(T_t) \geq v(N) - p(N)-\frac{1}{t}$. Since there are finitely many values for
$T_t$, some set $T$ must occur infinitely often. Taking the limit $t \rightarrow
\infty$ for this subsequence we get $v(T) - p(T) \geq v(N) - p(N)$, since
$N \in D(v;p)$, this must hold with equality. Therefore, by Lemma
\ref{lemma:characterization}, this is a Nash equilibrium.

Conversely, if $p \notin \pareto(F)$, we want to show that $p$ is not a Nash
equilibrium with $X(p) = N$. If $p \notin F$, then for some $T \subset N$ such
that $v(N) - p(N) < v(N \setminus T) - p(N \setminus T)$, so we can't have $X(p)
= N$. If $p \in F \setminus \pareto(F)$, then there is $p' \in F$ with $p' \geq
p$ and $i$ such that $p'_i >  p_i$. In particular, for all $T \ni i$, $p(T) <
p'_i + p(T\setminus i) \leq v(T \vert N \setminus T)$. That is, for all
such $T \ni i$, $v(N) - p(N) > v(N \setminus T) - p(N \setminus T)$, which
contradicts the first condition on Lemma \ref{lemma:characterization} and
therefore can't be an equilibrium.
\end{proof}

If $X(p)$ is not maximal, then Nash equilibria are not guaranteed to exist.
Consider for example the pricing game with one seller and one buyer with
valuation $v(\{1\}) = 1$ and a decision map such that $X(p_1) = \{1\}$ for $p_1 < 1$ and $X(p_1) = \emptyset$ otherwise.
Note that $X(p)$ is not maximal for $p_1=1$ as $\{1\}$ is also in the demand correspondence yet $X(1)=\emptyset$.
This game has no Nash equilibria, since for $p_1 < 1$
the deviation $p'_1 = \frac{1}{2}(1+p_1)$ is a strict best-response. For $p_1
\geq 1$, the deviation $p'_1 = 0.9$ is a strict best-response. This is the same
phenomenon that happens in first price auctions. Yet, similarly to first price
auctions, it is possible to show that given any decision map (which corresponds
to tie breaking rules in first price auctions), there exists an $\epsilon$-Nash
equilibrium for every $\epsilon > 0$.
{To prove that we first show that if $p$ is an equilibrium with respect
to some maximal tie breaking rule $X$,
then for any other tie breaking rule $X'$ there always exists $p^\epsilon$ that
is an $\epsilon$-Nash equilibria with respect to $X'$ and results with the same
welfare, $W(X'(p^\epsilon)) = W(X(p))$.}

\comment{
The following result allows us to export results on equilibria for maximal
decision maps to approximate equilibria of generic decision maps. Note that when
we say \emph{generic}, we still require that $X(p) \in \argmax v(S) - p(S)$, but
we don't require that $X(p)$ is maximal among the sets in $\argmax v(S) - p(S)$
anymore.
}

\begin{theorem}\label{lemma:any_map}
Fix a combinatorial valuation $v$. Let $X$ be a maximal decision map and $p \in
\nash^{X}$, i.e, $p$ is a pure Nash equilibrium with respect to the game defined by
$X$. Now, for any (not necessarily maximal) decision map $X'$, there are
$\epsilon$-Nash equilibria with respect to $X'$ converging to $p$, i.e.,
$\forall \epsilon > 0, \exists p^\epsilon \in \nash_\epsilon^{X'}$ such that
$W(X'(p^\epsilon)) = W(X(p))$ and $p^\epsilon \rightarrow p$.
\end{theorem}

The main idea of the proof is to construct for every $p \in \nash^{X}$ and
$\epsilon > 0$ a price vector $p^\epsilon$ such that $p^\epsilon_i = [p_i -
\frac{\epsilon}{n}]^+$ for $i \in X(p)$ and $p^\epsilon_i = p_i$ otherwise.
Then we argue that $p^\epsilon$ must be an $\epsilon$-Nash equilibrium of the
game induced by any decision map $X'$. We defer the details of the proof
\columnsversion{for the full version of this paper}{to the appendix}.
{We note that the combination of Theorem~\ref{cor:non_empty} and
Theorem~\ref{lemma:any_map} implies that:}

\begin{corollary}
For  any combinatorial valuation $v$, any decision map $X'$ and any $\epsilon >
0$, the set of welfare maximizing $\epsilon$-Nash equilibria is non-empty.
\end{corollary}

\section{Uniqueness of Equilibria}\label{sec:uniqueness}

In the previous section, we showed that if the decision map is maximal, equilibria
are guaranteed to exist for any combinatorial valuation. This equilibrium might
not be unique, as the Nash Bargaining Example (Example \ref{ex:nash_bargaining})
shows. In this section we show that submodularity of the valuation function is a
sufficient condition to guarantee uniqueness of utilities for pure Nash equilibrium.
Moreover, equilibrium prices are unique for every seller with non-zero utility.
For the subclass of gross-substitute valuation, the same claim also hold for
mixed Nash equilibrium.

\subsection{\columnsversion{Submodular valuations: uniqueness of \\ pure
equilibria}{Submodular valuations: uniqueness of pure
equilibria}}

We start with a few observations:

\begin{lemma}\label{lemma:pre_uniqueness}
If the buyer valuation is submodular, then for
any pure Nash equilibrium $p$, $p_i \geq v(i \vert N \setminus i)$.
Moreover, if $S = X(p)$ and $i \notin S$, then $v(i \vert S) = 0$.
\end{lemma}

\begin{proof}
Setting price $p_i < v(i \vert N \setminus i)$ makes seller $i$ always be
chosen, since for any $S \subseteq N \setminus i$, $v(S \cup i) - p(S \cup i) >
v(S) - p(S) + [v(i \vert S) - v(i \vert N)] \geq v(S) - p(S)$, where the last
step follows by submodularity. Therefore, no price $p_i < v(i \vert N \setminus
i)$ can be in equilibrium, since a deviation to any price between $p_i$ and $v(i
\vert N \setminus i)$ is an improvement.

Also, if $i \notin S = X(p)$, then $v(i \vert S) = 0$, otherwise he could
deviate to $p'_i = \frac{1}{2} v(i \vert S)$, get selected by the buyer and
have positive utility.
\end{proof}

The following simple fact about submodular functions will be useful in the following proof:

\begin{lemma}\label{lemma:submodular_fact}
If $v$ is a submodular function and $S \cap T = \emptyset$, then $v(T \vert S)
= 0$ iff $v(t \vert S) = 0$ for all $t \in T$
\end{lemma}

\begin{proof}
If $v(T \vert S) = 0$, then by monotonicity of $v$, $v(t \vert S) \leq v(T
\vert S) = 0$. Now, if  $v(t \vert S) = 0$ for all $t \in T$ then for let $T =
\{t_1, \hdots, t_k\}$ and  $$v(T \vert S) = \sum_{i=1}^k v(t_i \vert S \cup
\{t_j; j < i\}) \leq \sum_{i=1}^k v(t_i \vert S ) = 0 $$
\end{proof}

The previous lemma together with Lemma \ref{lemma:pre_uniqueness} imply that if
$v$ is submodular, then for any Nash equilibrium $p$ and $S = X(p)$, it must
holds that $v(S) = v(N)$. Looking again from the perspective of the welfare function, $W(p) = v(X(p))$,
the theorem above gives a \emph{Price of Anarchy} result for pricing games with
submodular function: is states that all Nash equilibria maximize the welfare, and thus the Price of Anarchy is exactly one.

\begin{theorem}
If the valuation $v$ is submodular, then the
utility of each seller $i$ in any pure Nash equilibrium is $u_i = v(i \vert N \setminus i)$.
In particular, for every pure Nash equilibrium profile $p$, if $v(i \vert N \setminus i) >
0$, then $p_i = v(i \vert N \setminus i)$ and $i\in X(p)$.
\end{theorem}

\begin{proof}
Let $p$ be a pure Nash equilibrium and $S = X(p)$. From Lemma
\ref{lemma:pre_uniqueness}, we know that $p_i \geq v(i \vert N \setminus i)$ for all $i$. Additionally for $i \notin S$, $v(i \vert S) = 0$
and therefore by sub-modularity $v(i \vert N \setminus i) = 0$. We are left to
prove that for $i \in S$, $p_i \leq v(i \vert N \setminus i)$. Fixing $i \in S$, we
use the first condition in Lemma \ref{lemma:characterization} to obtain a set $T
\not\ni i$ such that $v(S) - p(S) = v(T) - p(T)$. We consider two cases:

\emph{Case (i)} $T \not\subset S$. Observe that for all $j \in T \setminus S$,
$p_j = 0$, otherwise this seller could deviate to $\frac{1}{2} p_j$, making $T$
be selected by the buyer instead of $S$ and getting positive utility. Given that
since $S \in D(v;p)$:
$$v(S) - p(S) \geq v(S \cup T \setminus i) - p(S \cup T \setminus i) = v(S \cup
T \setminus i) - p(S \setminus i)$$
which implies that:
$$p_i \leq v(S) - v(S \cup T \setminus i) = v(N) - v(N \setminus i) = v(i \vert
N \setminus i)$$
where $v(S) = v(N)$ follows from Lemma \ref{lemma:pre_uniqueness}. The fact
that $v(S \cup T \setminus i) = v(N \setminus i)$ comes from the observation
that for $k \notin S \cup T$, $v(k \vert T) = 0$, otherwise seller $k$ could set
his price to $\frac{1}{2} v(k \vert T)$ and be selected, since $T \cup \{k\}$
would be preferable then $S$ at such prices. Since $v(k \vert T) = 0$, by
submodularity $v(k \vert S \cup T \setminus i) = 0$ and therefore  $v(S \cup T
\setminus i) = v(N \setminus i)$ by Lemma \ref{lemma:submodular_fact}.

\emph{Case (ii)} $T \subset S$. We have that $v(S) - p(S) = v(T) - p(T) \geq
\max ( v(S\setminus i) - p(S \setminus i), v(T \cup \{i\}) - p(T \cup \{ i
\}))$. In particular: $v(i \vert T) \leq p_i \leq v(i \vert S \setminus i)$.
Since by submodularity $v(i\vert T) \geq v(i \vert S \setminus i)$, the
inequalities must hold with equality, so $p_i = v(i \vert S \setminus i) = v(S)
- v(S\setminus i) = v(N) - v(N\setminus i) = v(i \vert N \setminus i)$. We know
$v(S) = v(N)$ from Lemma \ref{lemma:pre_uniqueness}. In order to see that
$v(S\setminus i) =  v(N\setminus i)$, notice that for any $j \notin S$ by the
second condition of Lemma \ref{lemma:characterization}:
$$v(S\setminus i) - p(S \setminus i) = v(S) - p(S) \geq v(S \cup j \setminus i)
- p(S \setminus i)$$
and therefore $v(j \vert S \setminus i) = 0$. Using Lemma
\ref{lemma:submodular_fact}, we get that $v(S \setminus i) = v(N \setminus i)$.
\end{proof}

We emphasize that even if $X(p)=S\neq N$, the utility of $i\in S$ is $v_i(i|N\setminus i)$ and not $v_i(i|S\setminus i)$. These two might be different, for example, this is the case for Bertrand competition.

The next example shows that submodularity is in some sense necessary. Weaker concepts of 'substitutability' such as XOS or subadditivity are not enough to ensure uniqueness of equilibria.

\begin{example}
Consider the following instance of the pricing game with three sellers and a
buyer with valuation $v$ such that $v(\emptyset) = 0$, $v(S) = 2$ for $1\leq
\abs{S} \leq 2$ and $v(S) = 3$ for $\abs{S} = 3$. This function is in XOS and is
subadditive, but it is not submodular, since $v(3 \vert \{1\}) = 0 < 1 = v(3
\vert \{1,2\})$.

The utilities in equilibrium are not unique, indeed, using the conditions in
Lemma \ref{lemma:characterization} it is easy to see that the pure equilibria
$p$ with $X(p) = N$ are given by:
$$p = (x,x,1-x) \text{ for } 0 \leq x \leq \textstyle\frac{1}{2}$$
and permutations thereof.
\end{example}

\subsection{Gross substitute valuations: uniqueness of mixed equilibria}

One other question that arises from the previous theorem is weather the
uniqueness in utilities also holds for mixed Nash equilibria. In what follows we
give a mixed Nash uniqueness result for the subclass of gross substitute
valuations. Recall that a valuation $v$ is gross substitute if
{when some prices are increased, there is always a demanded set which
contains all previously demanded items for which the price did not change.
Formally, $v$ is gross substitute if  }
for any price vectors $p' \geq p$ and $S \in D(v;p)$, there is $S' \in D(v;p')$ such that $S \cap \{j; p_j = p'_j\} \subseteq S'$. We say that a decision map $v$ is \emph{\textsc{Gs}-consistent} with a gross substitute valuation $v$ if given $p' \geq p$, $X(p) \cap \{j; p_j = p'_j\} \subseteq X(p')$.

It is known that when $v$ is a gross substitute valuation, then the
\emph{greedy algorithm} with lexicographic tie-breaking implements a
\textsc{Gs}-consistent decision map. This greedy algorithm starts with $S =
\emptyset$ and while $\max_{i \notin S} v(i \vert S) - p_i \geq 0$, adds to $S$
the lexicographically first element for which the maximum is reached. Let $X(p)$
be the resulting set from this process. It follows from
\cite{DressTerhalle_WellLayered, Murota96} that $X(p) \in \argmax_S v(S) -
p(S)$.

Before stating and proving the uniqueness result for gross substitutes, we give
a version of Lemma \ref{lemma:pre_uniqueness} for mixed Nash equilibria and a
technical lemma about gross substitute valuations.

\begin{lemma}\label{lemma:pre_uniqueness_mixed}
If the buyer valuation is submodular, then for
any mixed Nash equilibrium $p$, $p_i \geq v(i \vert N \setminus i)$ with
probability $1$. Moreover, if $\E[u_i(p)] = 0$, then $\Prob[v(i \vert X(p)) > 0]
= 0$.
\end{lemma}

\begin{proof}
The first part follows directly from the proof of Lemma
\ref{lemma:pre_uniqueness}. For the second part, if $\Prob[v(i \vert X(p)) > 0]
> 0$, then there is $\epsilon > 0$ for which $\Prob[v(i \vert X(p)) > \epsilon]
> 0$. Therefore, by deviating to price $p_i = \epsilon$, seller $i$ can be obtain
utility $\epsilon$ with positive probability, contradicting that $\E[u_i(p)] =
0$ for this equilibrium.
\end{proof}

\begin{lemma}\label{lemma:technical_gs}
Consider a gross substitute valuation function $v$ over $N$, disjoint sets $A,B$ and $j \notin A \cup B$. Assume that $v(i \vert A \cup j) = 0$  for all $i \in B$ and that $0 \leq v(j \vert A) - x < v(B \vert A)$ for some $x \geq 0$. Then there is $i \in B$ such that $v(j \vert A) - x < v(i \vert A)$.
\end{lemma}

\begin{proof}
Let $\epsilon = \frac{1}{2n}[v(B \vert A) - v(j \vert A) + x]$ and
define a price vector $p$ such that $p_j = x$, $p_i = 0$ for $i \in A$, $p_i = \epsilon$ for $i \in B$ and $p_i =\infty$ otherwise. Let $S$ be a maximal set in $D(v;p)$. Clearly $A \subseteq S$ since the elements in $A$ have zero price. Now, we argue that $S\setminus A \subseteq B$. We note that if $j \in S$, then no element in $B$ can be in $S$, since $v(i \vert A \cup j) = 0$ and $p_i > 0$ for all $i \in B$ while $p_i>0$. But $A \cup j$ is not optimal, since $v(B \vert A) - \epsilon \abs{B} > v(j \vert A) - x $. Therefore, $S$ must be a subset of $A \cup B$. Also, since $v(B \vert A) - \epsilon \abs{B} > v(j \vert A) - x \geq 0$, there must be some $i \in B \cap S$.

Consider such $i$ and the price vector $p'$ such that $p'_k = p_k$ for all $k \in A \cup \{i,j\}$ and $p'_k = \infty$ otherwise. By the definition of gross subtitutes, there must be $S' \in D(v,p')$ such that $i \in S'$. Since items in $A$ have zero price, we can assume that $A \subseteq S'$. Since $i$ has positive price and $v(i \vert A \cup j) = 0$, then $j$ can't be in $S'$. So, $S' = A \cup i$. This implies that $v(A \cup i) - p(A \cup i) \geq v(A \cup j) - p(A \cup j)$, which can be re-written as: $v(i \vert A) > v(i \vert A) - \epsilon \geq v(j \vert A) - x$.
\end{proof}

\begin{theorem}
Let the valuation $v$ be gross substitutes, $X$ be a \textsc{Gs}-consistent decision map and $p$ be a vector of independent random variables forming a mixed Nash equilibrium.
It holds that if $v(i \vert N \setminus i) > 0$ then $p_i = v(i \vert N \setminus i)$ and $i\in X(p)$ deterministically.
Additionally, if $v(i \vert N \setminus i) = 0$, then seller $i$ has expected utility of zero at any mixed Nash equilibrium, i.e., $\E[u_i(p)] = 0$.
\end{theorem}

\begin{proof}
We know by Lemma \ref{lemma:pre_uniqueness_mixed} that $p_i \geq v(i \vert N \setminus i)$ with probability $1$.
Let $S = \{i \in N;  \Prob [ p_i > v(i \vert N \setminus i) ] > 0\}$. If we show that for every $i \in S$, $\E[u_i(p)] = 0$, then we are done, since then clearly $v(i \vert N \setminus i) = 0$, otherwise this seller would be able to get positive utility by posting a price $\frac{1}{2} v(i \vert N \setminus i)$.

Let $p_i^t$ be some  price in the support of $p_i$ for which $\E[u_i(p)] = \E[u_i(p_i^t, p_{-i})]$, $\Prob[ p_i \geq p_i^t ] \geq \textstyle\frac{1}{t}$ and $\Prob[ p_i \leq p_i^t ] \geq 1-\textstyle\frac{1}{t}$. Define $p^t = (p_1^1, \hdots, p_n^t)$. Since $X(p^t)$ can take finitely many values, there is a set $T \subset N$ such that $T = X(p^t)$ infinitely often. We know that $S \not\subseteq T$, since for very large $t$, $p_i^t > v(i \vert N \setminus i)$ for all $i \in S$.

For any $i \notin T$ we know that
$$\E[u_i(p)] = \E[u_i(p_i^t, p_{-i})] = p_i^t \cdot \Prob[ i \in X(p_i^t, p_{-i})]$$
where the first equality comes from  $p_i^t$ being a best response to $p_{-i}$. Now, since $X$ is a \textsc{Gs}-consistent decision map, if $p_{-i} \leq p^t_{-i}$, then $i \notin X(p_i^t, p_{-i})$, therefore:
$$\begin{aligned}\Prob[ i \in X(p_i^t, p_{-i})] & \leq \Prob[ \exists j \neq i ; p_{j} > p_{j}^t ] \leq \\ & \leq \textstyle\sum_{j \neq i} \Prob[ p_{j} > p_{j}^t ] \leq \textstyle\frac{n-1}{t}\end{aligned}$$
Taking $t \rightarrow \infty$ we get that $\E[u_i(p)] = 0$.

Now we claim that  $v(i \vert T) = 0$. If not, then the price $p'_i = \frac{1}{2} v(i \vert T)$ would guarantee that: $i \in X(p')$ for $p' = (p'_i, p^t_{-i})$ since $v(T \cup \{i\}) - p'(T \cup \{i\}) > v(T) - p'(T) = \max_{S'; i \notin S'} v(S') - p'(T')$. Therefore, by the fact that the valuations are gross substitutes:
$$\begin{aligned}
\E[u_i(p)] & \geq \E[u_i(p'_i, p_{-i})] = p'_i \cdot \Prob[i \in X(p'_i, p_{-i})]\\
& \geq p'_i \cdot \Prob[p_{j} \geq p^t_{j}, \forall j \neq i] \geq \frac{1}{t^{n-1}}> 0\end{aligned}$$
which contradicts that $\E[u_i(p)] = 0$. So, it must be the case that $v(i \vert T) = 0$. In particular, $v(N) = v(T)$ by Lemma \ref{lemma:submodular_fact}.

Now, in order to complete the proof, we want to show that $S \cap T = \emptyset$, since this implies that if for some seller $j$, $\Prob [ p_j > v(j \vert N \setminus j) ] > 0$, then $j \notin T$ and therefore $\E[u_j(p)]= 0$. Assume for contradiction that there is $j \in S \cap T$ and take $t$ large enough we can assume that $p^t_j > v(j \vert N \setminus j)$. Then:
$$v(j \vert T\setminus j) - p_j^t < v(j \vert T\setminus j) - v(j \vert N \setminus j) = v(N \setminus T \vert T \setminus j)$$
since $v(T) = v(N)$.
\comment{
Since
$$\begin{aligned}
v(j \vert T \setminus j) - v(j \vert N \setminus j) &= v(T) - v(T \setminus j) - v(N) + v(N \setminus j)  \\ &= v(N \setminus j) - v(T \setminus j) = v( N\setminus T \vert T \setminus j)
\end{aligned}$$}
This allows us to apply Lemma \ref{lemma:technical_gs} with $j$, $A = T \setminus j$, $B = N \setminus T$ and $x = p_j$, we get that there is $i \notin T$ such that: $v(j \vert T \setminus j) - p_j^t < v(i \vert T \setminus j)$. Therefore, if such $i$ deviates to price $p'_i > 0$ such that $v(j \vert T \setminus j) - p_j^t < v(i \vert T \setminus j) - p'_i$, then $i \in X(p')$ for $p' = (p'_i, p^t_{-i})$, since $v(T \cup i \setminus j) - p'(T \cup i \setminus j) > v(T) - p'(T) \geq \max_{S' \not\ni i} v(S') - p'(S')$. Therefore:
$$\E[u_i(p)] \geq \E[u_i(p'_i, p_{-i})] \geq p'_i \cdot \Prob[p_{j} \geq p^t_{j}, \forall j \neq i] > 0$$
contradicting that $\E[u_i(p)]=0$. This shows that there can't be a seller $j$ in $S \cap T$, concluding the proof.
\end{proof}

\section{Extensions of the basic model}\label{sec:extensions}

In this section we explore some natural generalizations of the basic
model. We first consider costly services, then we consider valuation uncertainty
and the connection to the model of multiple buyers with no price discrimination.
Finally, we consider the case that a single seller controls multiple services
but can only price individual services.

\subsection{Service costs}
\label{sec:costs}

In the basic model we assumed that the seller has no cost in providing the
service.
Now, we consider the extension in which each seller $i$ has 
cost $c_i \geq 0$ for providing the service.
We model
this by changing the utility of each seller $i$ to:
$$u_i(p) = (p_i - c_i) \cdot \one\{ i \in X(p)\}$$
Clearly in this new game, it is a dominated strategy to post a price $p_i <
c_i$. The first thing we notice is that maximality of the map is not enough to
guarantee existence of pure Nash equilibrium anymore. Consider, the
following example:

\begin{example}
Consider a game with two sellers and a buyer with valuation $v(S) = 3$
if $\abs{S} \geq 1$ and $v(\emptyset) = 0$, and two sellers with costs
$c_1 = 1$ and $c_2 = 2$. Also, let the seller break ties in favor of the
costlier seller, i.e., $X(p) = \{i\}$ if $p_i < p_{3-i} \leq 3$, $X(p) = \{2\}$
if $p_1 = p_2 \leq 3$ and $X(p) = \emptyset$ otherwise. The map $X$ is maximal,
and yet there is no pure Nash equilibrium: there is no equilibria with $\min\{p_1, p_2\} > 3$,
since one of the sellers can decrease his price and be selected. For the case $\min\{p_1, p_2\} \leq 3$,  there is no equilibria with
$p_1 \neq p_2$  since the seller with lowest price can improve his utility by slightly increasing his price. We are left with $p_1 = p_2 \leq 3$. By the tie breaking rule, $X(p) =
\{2\}$ and since $c_2 = 2$, we must have $p_1 = p_2 \geq 2$, otherwise $2$ would
be getting negative utility. This can't be an equilibrium since $p_1$ can
decrease his price and be selected.
\end{example}

However, we show that with an additional mild assumption on the decision map $X$, we
can prove the existence of $\epsilon$-Nash equilibria for any $\epsilon > 0$.
We say that a map is \emph{up-consistent} if for any price vector $p$ if $X(p)
= S$ and $p'_i > p_i$, then either $X(p'_i, p_{-i}) = S$ or $i \notin X(p'_i,
p_{-i})$. Lexicographically breaking ties is up-consistent as we show next.

\begin{lemma}
The map that chooses the lexicographically first set among the (maximal) sets of
$D(v;p)$ is (maximal) up-consistent.
\end{lemma}

\begin{proof}
Let $S = X(p)$, $i \in S$, $p'_i > p_i$ and $p'=(p'_i, p_{-i})$. We want to show
that if $i \in X(p')$, then $X(p') = S$. By the
lexicographic rule, if $X(p') \neq S$, there must be $T \in D(v;p') \setminus
D(v;p)$ such that $i \in T$. Given that $T \notin D(v;p)$, there must be a set
$T'$ such that $v(T) - p(T) < v(T') - p(T')$. Since $p$ coincides with $p'$
except from coordinate $i$ and $i \in T$, then: $v(T) - p'(T) < v(T') -
p'(T')$, therefore $T$ can't be in $D(v;p')$.
\end{proof}

Using the definition of up-consistency, we show the existence of
$\epsilon$-Nash for any $\epsilon > 0$. Moreover, the following theorem also
states that there are $\epsilon$-Nash equilibria with optimal welfare, where the
welfare function defined as the sum of the utilities of the agents (the sellers
and the buyer): $W(p) = v(X(p)) - c(X(p))$. This implies a Price of Stability
of one.

\begin{theorem}
Consider any combinatorial valuation $v$, a vector of service costs $c$ and an
up-consistent decision map $X$. For every $\epsilon > 0$, there exists an $\epsilon$-Nash equilibrium $p$ with $X(p) = X(c)$.
\end{theorem}

\begin{proof}
 Fix $\epsilon > 0$, let $S = X(c)$ and define:
 $$F = \left\{p \in \R_+^{N} ;   p \geq c, X(p) = S \right\}$$
 Clearly $c \in F$. Set initially $p = c$ and while there exists some $i \in S$
such that $i \in X(p'_i, p_{-i})$, for some $p'_i \geq p_i + \epsilon$, update
$p_i$ to $p'_i$. By up-consistency we mantain the invariant that $X(p) = S$
during this
procedure. When it stops, we have a vector $p$ such that $X(p) = S$, for all
$i \in S$ and $p'_i \geq p_i + \epsilon$, $X(p'_i) \not\ni i$ and for $i \notin
S$, $p_i = c_i$.

 Now, it is easy to see that this is an $\epsilon$-Nash equilibrium: for $i \in
S$, $u_i = p_i$. In order to increase his utility by $\epsilon$, each seller
needs to deviate to a price $p'_i \geq p_i + \epsilon$, but from the
construction he won't be allocated at that price. For $i \notin S$, we know
that for any set $T \ni i$, $v(S) - p(S) \geq v(T) - p(T)$, so if $i$ deviated
to any price $p'_i > c_i$, then for $p' = (p'_i, p_i)$ we would have $v(S) -
p'(S) > v(T) - p'(T)$, hence $i \notin X(p')$.
\end{proof}

The following result characterizes the sets picked in equilibrium:

\begin{lemma}
If $p$ is a pure Nash equilibrium (or the limit of $\epsilon$-Nash equilibria as $\epsilon \rightarrow 0$), then:
 $$v(i \vert S \setminus i) \geq c_i, \forall i \in S$$
 $$v(T \cup j) + c(S \setminus T) - c_j \leq v(S), \forall T \subseteq S \text{ and } j \notin S$$
\end{lemma}

\begin{proof}
The first part follows from the fact that $S \in D(v;p)$ and therefore $v(S) -
p(S) \geq v(S \setminus i) - p(S \setminus i)$ and the fact that $p_i \geq c_i$.
The second statement comes from the fact that for $j \notin S$ and $p'_j > c_j$,
$i \notin X(p'_j, p_{-j})$, in particular: $v(S) - p(S) \geq v(T \cup j) - p(T)
- p'_j$. Taking $p'_j \rightarrow c_j$ and using the fact that $p_i \geq c_i$,
we get the desired result.
\end{proof}

In particular, $v(S) - c(S) \geq \max\{ v(S \setminus i) - c(S \setminus i),
v(S \cup j) - c(S \cup j), v(S \cup j \setminus i) - c(S \cup j \setminus i),
\forall i \in S, j \notin S\}$. This corresponds to a minimum of the local
search procedure that seeks to optimize $v(S) - c(S)$ by either adding an
element, removing an element or swapping an element in the set by an element
outside. Gul and Stachetti \cite{GulStachetti} show that if a valuation function
$v$ is gross subtitutes, then this procedure doesn't get stuck in suboptimal
local maxima. In particular, this implies that:

\begin{corollary}
If the valuation $v$ is gross substitutes and $p$ is a Nash equilibrium (or a
limit of $\epsilon$-Nash equilibria), then the welfare of the allocation $W(p) =
v(X(p)) - c(X(p))$ is optimal, i.e., $W(p) = \max_{S \subseteq N} v(S) - c(S)$.
In other words, the Price of Anarchy of this game, for gross substitute
valuations is $1$.
\end{corollary}

In the following example we show that submodularity is not enough to guarantee
that all Nash equilibria have optimal welfare, unlike the case without costs.
In fact, the welfare of a Nash equilibrium might be arbitrarily smaller then the
optimal welfare, i.e., the Price of Anarchy is
unbounded.

\begin{example}
\label{ex:unbounded_poa_submodular}
{Let $k$ and $\ell$ be integers with $k \gg \ell$ and $k$ odd. Now,
consider a pricing game with $k + \ell+1$ sellers indexed by $\{0, 1,
\hdots, k+\ell \}$ with costs $c_i = k-\ell$ for $i=0,...,\ell$ and $c_i = 0$ for
$i=\ell+1,...,\ell+k$. Define the valuation of the buyer as follows: $v(S) =\min
\{k\cdot \ell, \sum_{i \in S} w_i \}$ where $w_i = k$ for $i=0,...,\ell$ and $w_i =
\ell$ for $i=\ell+1,...,\ell+k$. Assume that the buyer breaks ties
lexicographically.}

{Consider now the following price vector: $p_i = k-\ell$
for $i=0,...,\ell$ and $p_i = \ell$ for $i=\ell+1,...,\ell+k$.
First we claim that this price vector is a Nash equilibrium. By the
lexicographic rule, $X(p) =\{0, \hdots, \ell-1\}$. For any seller $i \in X(p)$,
if $i$ increases his price, the buyer will replace him by seller $\ell$, choosing the set $\{0,1,
\hdots, \ell\} \setminus i$. For seller $\ell$, he cannot become selected by
increasing his price and if he decreases his price he gets negative utility
since his cost is $k -\ell$. Now, for any seller $i \in \{\ell+1, \hdots,
\ell+k\}$, they are not selected even if they decrease their price to zero.
Now, notice that $W(p) = v(X(p)) - c(X(p)) = \ell \cdot \ell$. For the set
$S^* = \{\ell+1, \hdots, \ell+k\}$, $v(S^*) -
c(S^*) = \ell \cdot k$. As $k/\ell \rightarrow \infty$ the gap
between the optimal welfare and the welfare of this Nash equilibrium goes to
infinity.}
\end{example}

\comment{
\begin{example}
 {Let $k$ be a large odd integer and $\ell= \frac{1}{2}(k+1)$. Now,
consider a pricing game with $k + \ell+1$ sellers indexed by $\{0, 1,
\hdots, k+\ell \}$ with costs $c_i = k-\ell$ for $i=0,...,\ell$ and $c_i = 0$
for
$i=\ell+1,...,\ell+k$. Define the valuation of the buyer as follows: $v(S) =\min
\{k\cdot \ell, \sum_{i \in S} w_i \}$ where $w_i = k$ for $i=0,...,\ell$ and
$w_i =
\ell$ for $i=\ell+1,...,\ell+k$. Assume that the buyer breaks ties
lexicographically.}

{Consider now the following price vector: $p_i = k-\ell$
for $i=0,...,\ell$ and $p_i = \ell$ for $i=\ell+1,...,\ell+k$.
First we claim that this price vector is a Nash equilibrium. By the
lexicographic rule, $X(p) =\{0, \hdots, \ell-1\}$. For any seller $i \in X(p)$,
if $i$ increases his price, the buyer will replace him by seller $\ell$,
choosing the set $\{0,1,
\hdots, \ell\} \setminus i$. For seller $\ell$, he cannot become selected by
increasing his price and if he decreases his price he gets negative utility
since his cost is $k -\ell$. Now, for any seller $i \in \{\ell+1, \hdots,
\ell+k\}$, they are not selected even if they decrease their price to zero.
Now, notice that $W(p) = v(X(p)) - c(X(p)) = \ell \cdot \ell =
\frac{1}{4}(k+1)^2$. for the set $S^* = \{\ell+1, \hdots, \ell+k\}$, $v(S^*) -
c(S^*) = \ell \cdot k = \frac{1}{2}k(k+1)$. As $k \rightarrow \infty$ the gap
between the optimal welfare and the welfare of this Nash equilibrium goes to
$2$.}
\end{example}
}

We also note that even when equilibria exist, it is not necessarily unique, even
when the valuations are submodular. Consider the following example:

\begin{example}
Consider the game with three sellers $\{1,2,3\}$ with costs $0.1$, $0.1$ and
$0.3$ and a buyer with the following submodular valuation $$v(S) = \min\{2,
\one\{1\in S\} +   \one\{2\in S\} + 2 \cdot \one\{3\in S\}  \}$$
Also, let $X$ be the decision map that picks the lexicographically first element
in $D(v;p)$. Now, note that all the vectors $(p_1, p_2, 0.3)$ with $p_1 \geq
0.1$, $p_2 \geq 0.1$ and $p_1 + p_2 \leq 0.3$ are Nash equilibria.
\end{example}

\subsection{Value uncertainty and Multiple buyers}\label{subsec:value_uncertainty}

The basic pricing game assumes that the valuation function of the buyer is public information. We consider here the case where sellers have uncertainty about the true valuation of the buyer, i.e., sellers only know that $v$ is drawn from a certain distribution
$\mathcal{D}$. In such case each seller $i$ seeks to maximize:
$$u_i(p) = \E_{v \sim \mathcal{D}}[p_i \cdot \one\{ i \in X_v(p) \}]$$

If $\mathcal{D}$ is the uniform distribution, this is equivalent to the scenario where there are multiple buyers, each buyer $k\in [m]$ with a valuation function $v_k:2^{N} \rightarrow \R_+$ and sellers are not allowed to price discriminate, i.e, they need to post the same price for all the buyers. Upon facing a price vector $p$, each buyer  purchases the bundle $X_k(p) \in D(v_k, p)$ and the revenue of seller $i$ is:
$$u_i(p) = \textstyle\sum_{k=1}^m p_i \cdot \one \{ i \in X_k(p) \}$$

We observe that the above model is a generalization of \emph{Bertrand networks}
that were defined by Babaioff, Lucier and Nisan \cite{bertrand_networks}. A
Bertrand Network is a game defined on a graph where each node is a seller and
his strategy is to set a price $p_i \in [0,1]$. Each edge corresponds to a set of buyers each interested in buying one item for price at most $1$ and
who chooses to purchase from the accessible seller (incident node) offering the
cheaper price, breaking
ties in a fixed but arbitrary manner. This defines a game among the sellers
whose utility is given by their total revenue, i.e., the price posted multiplied
by number of buyers that decide to purchase from this seller. This naturally maps to an
instance of our pricing game, where each node corresponds to a seller and each
edge $e = (i,j)$ corresponds to a buyer $k$ whose valuation is $v_k(S_k) = 1$ if $S_k \cap
\{i,j\} \neq \emptyset$ and zero otherwise.

Babaioff et al \cite{bertrand_networks} show that pure Nash equilibria might
not exist, but mixed Nash equilibria always exist. Both results carry over to
our setting. The non-existence of pure Nash for some instances follows directly,
since Bertrand Networks is a particular case of our pricing game. The existence
of mixed Nash equilibria proof in \cite{bertrand_networks} is non-trivial and it
\emph{doesn't} follow from Nash's Theorem, as usual, since the strategy space is
infinite and the utility functions are discontinuous. The strategy used for
proving existence of equilibrium, follows by applying a general result by Simon
and Zame \cite{simon_zame}. The same type of technique yields the existence of
mixed Nash in our setting as well. We defer the proof to the full version of this
paper.

We remark that the example in \cite{bertrand_networks} can also be used to show
an instance of the game with multiple buyers without $\epsilon$-Nash equilibria
for small values of $\epsilon$:

\begin{example}
Consider a setting with two sellers $\{A,B\}$ and two buyers with valuations
$v_1(S_1) = \one\{A \in S_1\}$ and $v_2(S_2) = \one\{ S_2 \neq \emptyset\}$. We show
that for sufficiently small $\epsilon$,  no (pure) $\epsilon$-Nash exist
for any tie breaking rule. \comment{Therefore, if $p_1 \leq p_2 \leq 1$, then: $u_1(p) =
2p_1$ and $u_2(p) = 0$. For $p_2 < p_1 \leq 1$, $u_1(p) = p_1$, $u_2(p) =
p_2$.}

Fix some value of $\epsilon$ such that $0 < \epsilon < \frac{1}{10}$ and assume
that $p$ is an $\epsilon$-Nash equilibrium. First we note that it must be the
case that $u_1(p) \geq 1-\epsilon$ since for any $p_2$, $u_1(1,p_2) \geq 1$.
Now, consider two cases:

Case (i) $p_1 \geq 1-\epsilon$. If $p_2 \geq
p_1$, then $u_2(p) = 0$ and seller $2$ can deviate to $p_1 - \epsilon$  giving
him utility: $u_2(p_1, p_1-\epsilon) = p_1-\epsilon \geq 1-2\epsilon >
\epsilon$. If $p_1 \geq p_2 \geq p_1 - \epsilon$, then seller $1$ can deviate to
$ p_2 - 2\epsilon$ and get utility $u_1(p_2 - 2\epsilon, p_2) = 2(p_2 -
2\epsilon) \geq 2-8\epsilon \geq p_1 + (1-8\epsilon) > p_1 + \epsilon$. Now, if
$p_2 < p_1 - \epsilon$, then seller $2$ can increase his utility by raising his
price by $\epsilon$ and still be chosen.

Case (ii): $p_1 < 1-\epsilon$. If $p_2 < p_1$ or $p_1 <
\frac{1-\epsilon}{2}$, then $u_1(p) < 1-\epsilon$ and player $1$ can improve his
utility by more then $\epsilon$ by raising his price to $1$. If
$\frac{1-\epsilon}{2} \leq p_1 \leq p_2$, then $u_2(p) = 0$, but seller $2$ can
lower his price to $p_1 - \epsilon$, getting utility at least
$\frac{1-\epsilon}{2} - \epsilon \geq \epsilon$.
\end{example}

\subsection{Multiple services per seller}

\comment{
\mbcomment{We need to present POA and POS results that are independent of
computation, and then remark that the "positive" side, of achieving O(log n)
approximation is also poly-time computable.}}

Finally, we relax the assumption that each seller controls a
single service. Let $N$ be the total set of services and let $N = \cup_{i=1}^r
N_i$ be a disjoint partition ($N_i \cap N_{i'} =
\emptyset$ for $i \neq i'$) where $N_i$ represent the services controlled by
seller $i$. Given a buyer with valuation $v:2^N \rightarrow \R$ and a decision
map $X$, this defines a game between sellers where the strategy of each seller
is to set prices for each service in $N_i$, i.e., $p_i \in \R_+^{N_i}$. The
utility of seller $i$ is then given by:
$$u_i(p) = \textstyle\sum_{j \in N_i} p_j \cdot \one \{ j \in X(p) \}$$
Observe that with two goods that are perfect substitutes but controlled by the
same seller, their price won't go down to zero as in Bertrand competition, since
a seller won't be competing with himself.

We illustrate the subtleties of this variant by analyzing the special case of
the \emph{monopolistic seller}, i.e., where one seller (seller $1$) controls all
the services and posts prices for each of them individually. This boils down to
an optimization problem where the goal of the seller is to find a price vector
$p$ that optimizes his utility:
 $$u_1(p):=\sum_{j \in N} p_j \cdot \one \{ j \in X(p) \}$$

{Note that the constraint that the seller prices {\em individual services} and
not bundles might prevent the seller, although a monopolist,  from extracting
the entire surplus.}

First we show the case that for the case of submodular valuations, this can be
phrased as an optimization over sets:

\begin{lemma}\label{lemma:monopolistic_seller_format}
If the valuation $v$ is submodular and $X$ is maximal, then
$\max_{p \in \R^N_+} u_1(p) = \max_{S \subseteq N} \sum_{i \in S} v(i \vert S
\setminus i)$
\end{lemma}

\begin{proof}
Given $p \in \R^N_+$, let $S = X(p)$. Since $S \in D(v;p)$ it must be the case
that $p_i \leq v(i \vert S \setminus i)$, so $u_1(p) = \sum_{i \in S} p_i \leq
\sum_{i \in S} v(i \vert S \setminus i)$. Conversely, given $S \subseteq N$,
setting prices $p_i = v(i \vert S \setminus i)$ for $i \in S$ and $p_i = \infty$
otherwise, we get $X(p) = S$ and $u_1(p) = \sum_{i \in S} v(i \vert S \setminus
i)$.
\end{proof}

{Since there are no costs, the welfare of this game can be defined as
$W(p) = v(X(p))$. We note that unlike the case where each seller controls one
good, the welfare in equilibrium can be far from efficient even if the buyer's
valuation is gross substitutes. Let $H_k = \sum_{j=1}^k \frac{1}{j} =
\Theta(\log n)$ be the $k$-th  harmonic number and define the valuation $v(S)$ such
that: $v(S) = 1+\epsilon$ if $\abs{S} = 1$ and $v(S) = H_{\abs{S}}$ if $\abs{S}
\geq 2$. This function is clearly submodular (in fact, it is also gross
substitutes). By Lemma \ref{lemma:monopolistic_seller_format}, the revenue in
a Nash equilibrium is of the format $\sum_{i \in S} v(i \vert S \setminus i)$
for $S = X(p)$, which is maximized for sets $S$ of size $1$. Therefore, for all
equilibria, the buyer buys at most one item, generating $1+\epsilon$ welfare,
while the optimal achievable welfare is $H_n$. Since this
is the unique equilibrium, this implies an logarithmic lower bound on the Price
of Anarchy and Price of Stability.}

{In what follows, we show that this bound is tight, i.e, for all Nash
equilibria the welfare is at least a $\Omega(1/\log n)$ fraction of the optimal
welfare. In other words, this implies an upper bound of $\log n$ of the Price of
Anarchy (and thus also on the Price of Stability).
We show that by proving that there is always a set $S$ such that
$\sum_{i \in S} v(i \vert S \setminus i) \geq \frac{1}{H_n} v(N)$. Therefore,
the seller can always secure that much revenue. Since the seller revenue is a
lower bound to the welfare we have:}

\begin{theorem}
If the valuation $v$ is submodular and $X$ is maximal, then for every pure Nash
equilibrium $p$ of the monopolistic seller game\footnote{since it is a
one-player
game, a Nash equilibrium corresponds to an optimal price vector}, $v(X(p)) \geq
\frac{1}{H_n} v(N)$. In other words, the Price of Anarchy is bounded by $H_n =
O(\log n)$. Moreover, there is a randomized polynomial-time algorithm that
finds a vector $p$ with $u_1(p) \geq  v(N) /2H_n $ with constant probability.
\end{theorem}

This is a consequence of the following lemma:

\begin{lemma}
\label{lemma:sample}
Given a function $v:2^N \rightarrow \R_+$, consider the following randomized
algorithm: pick a size $k \in \{1, \hdots, n\}$ with probability $(k \cdot
H_n)^{-1}$ and then pick a random set $S$ of size $k$. Then: $\E \left[
\sum_{i\in S} v(i \vert S \setminus i) \right] = \frac{1}{H_n} v(N)$.
\end{lemma}

\begin{proof}
Observe that we can rewrite the expectation as:
$$\E \left[ \sum_{i\in S} v(i \vert S \setminus i) \right] = \sum_{k=1}^n
\frac{1}{k H_n} \E_{\abs{S} = k} \left[ k \cdot v(S) - \sum_{i \in S} v(S
\setminus i)\right]$$
Let $\tilde{v}(k) = \E_{\abs{S} = k} v(S)$. We note that:
$$\begin{aligned}& \E_{\abs{S} = k} \sum_{i \in S} v(S \setminus i) =
\sum_{S:\abs{S} = k} \frac{1}{{n \choose k}} v(S\setminus i) \\ & \quad =
\sum_{T:\abs{T} = k-1} \frac{n-k+1}{{n \choose k}} v(T)  =
\sum_{T:\abs{T} = k-1} \frac{k}{{n \choose k-1}} v(T) = k \tilde{v}(k-1)
\end{aligned}$$
Therefore:
$$H_n \cdot \E \left[ \sum_{i\in S} v(i \vert S \setminus i) \right] =
\sum_{k=1}^n \tilde{v}(k) - \tilde{v}(k-1) = \tilde{v}(n) = v(N)$$
\end{proof}

This in particular implies a polynomial time randomized approximation algorithm
for the optimization problem faced by the seller. We remark that since $u_1(p)$
is bounded from above by $v(N)$, we can get expected utility of
$\Omega(v(N)/H_n)$ with high probability by running the algorithm $\Theta(\log
n)$ and taking the best output. We make this statement precise in following
Lemma:

\begin{lemma}\label{lemma:high_prob}
Given a submodular function $v$ and $r(S) := \sum_{i \in S} v(i \vert S
\setminus i)$, then if $S_1, \hdots, S_k$ are independent samples according to
the procedure described in Lemma~\ref{lemma:sample} for $k= s \cdot H_n$,
then $\Prob[\max_i r(S_i) \geq \frac{v(N)}{2 H_n}] \geq  1-e^{-s/2}$.
\end{lemma}
\begin{proof}
Since $v$ is a submodular function $r(S) = \sum_{i \in S} v(i \vert S \setminus
i) \leq v(N)$ for all $S$. Therefore, for all $t \in [0, v(N)]$
$$\frac{v(N)}{H_n} = \E[r(S_i)] \leq t \cdot \Prob\left[ r(S_i)
 < t \right] + v(N) \cdot \Prob\left[ r(S_i) \geq t  \right]$$
so $\Prob\left[ r(S_i) \geq t  \right] \geq \frac{\frac{v(N)}{H_n} -t }{v(N) -
t} \geq \frac{1}{H_n} - \frac{t}{v(N)}$. Taking $t = \frac{v(N)}{2 H_n}$ we get
$\Prob \left[ r(S_i) \geq \frac{v(N)}{2 H_n}  \right] \geq \frac{1}{2 H_n} $.
Therefore, $$\Prob \left[ \max_i r(S_i) \geq \frac{v(N)}{2 H_n}  \right] \geq
1-(1-\frac{1}{2H_n})^{s H_n} \geq 1-e^{-s/2}$$
\end{proof}

Finally, we note that the optimization problem faced by the monopolistic seller
is a particular case of the \emph{Unique Coverage Problem} studied by
Demaine et al \cite{Demaine08}: given an universe set $U$ and subsets $Y_1, Y_2,
\hdots, Y_n$, find a collection of subsets that maximizes the number of elements
covered by \emph{exactly} one set. In other words, find $S \subseteq [n]$ in
order to maximize $\sum_{i \in S} \abs{Y_i \setminus \cup_{j \in S \setminus i}
Y_j}$. This is exactly the optimization problem faced by the
monopolistic seller when $v(S) = \abs{\cup_{i \in S} Y_i}$. Demaine et al
\cite{Demaine08} give a $O(\log n)$ approximation of the Unique Coverage problem
and show an $\Omega(\log^\sigma n)$ hardness of approximation for some constant
$\sigma > 0$ under suitable complexity assumptions.
This in particular implies logarithmic hardness of approximation for the monopolistic seller problem.



\section*{Acknowledgements}

We thank Balasubramanian Sivan for many comments on this
manuscript, specially for helping us to improve our Example
\ref{ex:unbounded_poa_submodular}.

\bibliographystyle{abbrv}
\bibliography{sigproc}

\begin{thebibliography}{10}

\bibitem{wsj_article}
J.~Angwin and D.~Mattioli.
\newblock Coming soon: Toilet paper priced like airline tickets.
\newblock \url{
  http://online.wsj.com/article/SB10000872396390444914904577617333130724846.ht%
ml}, 2012.

\bibitem{frugal0}
A.~Archer and {\'E}.~Tardos.
\newblock Frugal path mechanisms.
\newblock In {\em SODA}, pages 991--999, 2002.

\bibitem{bertrand_networks}
M.~Babaioff, B.~Lucier, and N.~Nisan.
\newblock Bertrand networks.
\newblock In {\em ACM Conference on Electronic Commerce}, pages 33--34, 2013.

\bibitem{BernheimWhinston}
D.~Bernheim and M.~D. Whinston.
\newblock Menu auctions, resource allocation, and economic influence.
\newblock {\em The Quarterly Journal of Economics}, 101(1), 1986.

\bibitem{bertrand}
J.~Bertrand.
\newblock Book review of theorie mathematique de la richesse sociale and of
  recherches sur les principles mathematiques de la theorie des richesses.
\newblock {\em J des Savants}, 67(2):499--508, 1883.

\bibitem{BR11}
K.~Bhawalkar and T.~Roughgarden.
\newblock Welfare guarantees for combinatorial auctions with item bidding.
\newblock In {\em SODA}, pages 700--709, 2011.

\bibitem{ChawlaN09}
S.~Chawla and F.~Niu.
\newblock The price of anarchy in bertrand games.
\newblock In {\em ACM Conference on Electronic Commerce}, pages 305--314, 2009.

\bibitem{ChawlaR08}
S.~Chawla and T.~Roughgarden.
\newblock Bertrand competition in networks.
\newblock In {\em SAGT}, pages 70--82, 2008.

\bibitem{frugal3}
N.~Chen, E.~Elkind, N.~Gravin, and F.~Petrov.
\newblock Frugal mechanism design via spectral techniques.
\newblock In {\em FOCS}, pages 755--764, 2010.

\bibitem{CKS08}
G.~Christodoulou, A.~Kov{\'a}cs, and M.~Schapira.
\newblock Bayesian combinatorial auctions.
\newblock In {\em ICALP (1)}, pages 820--832, 2008.

\bibitem{Demaine08}
E.~D. Demaine, U.~Feige, M.~Hajiaghayi, and M.~R. Salavatipour.
\newblock Combination can be hard: Approximability of the unique coverage
  problem.
\newblock {\em SIAM J. Comput.}, 38(4):1464--1483, 2008.

\bibitem{DressTerhalle_WellLayered}
A.~Dress and W.~Terhalle.
\newblock Well-layered maps -- a class of greedily optimizable set functions.
\newblock {\em Applied Mathematics Letters}, 8(5):77 -- 80, 1995.

\bibitem{eisen_blog}
M.~Eisen.
\newblock Amazon's \$23,698,655.93 book about flies.
\newblock \url{http://www.michaeleisen.org/blog/?p=358}, 2011.

\bibitem{frugal2}
E.~Elkind, L.~A. Goldberg, and P.~W. Goldberg.
\newblock Frugality ratios and improved truthful mechanisms for vertex cover.
\newblock In {\em ACM Conference on Electronic Commerce}, pages 336--345, 2007.

\bibitem{frugal1}
E.~Elkind, A.~Sahai, and K.~Steiglitz.
\newblock Frugality in path auctions.
\newblock In {\em SODA}, pages 701--709, 2004.

\bibitem{Feldman13}
M.~Feldman, H.~Fu, N.~Gravin, and B.~Lucier.
\newblock Simultaneous auctions are (almost) efficient.
\newblock In {\em STOC}, pages 201--210, 2013.

\bibitem{GulStachetti}
F.~Gul and E.~Stacchetti.
\newblock Walrasian equilibrium with gross substitutes.
\newblock {\em Journal of Economic Theory}, 87(1):95--124, July 1999.

\bibitem{HKMN11}
A.~Hassidim, H.~Kaplan, Y.~Mansour, and N.~Nisan.
\newblock Non-price equilibria in markets of discrete goods.
\newblock In {\em ACM Conference on Electronic Commerce}, pages 295--296, 2011.

\bibitem{immorlica05}
N.~Immorlica, D.~R. Karger, E.~Nikolova, and R.~Sami.
\newblock First-price path auctions.
\newblock In {\em ACM Conference on Electronic Commerce}, pages 203--212, 2005.

\bibitem{immorlica10}
N.~Immorlica, D.~R. Karger, E.~Nikolova, and R.~Sami.
\newblock First-price procurement auctions.
\newblock Manuscript, 2010.

\bibitem{KelsoCrawford}
A.~Kelso and V.~Crawford.
\newblock Job matching, coalition formation, and gross substitutes.
\newblock {\em Econometrica}, 50(6):1483--1504, November 1982.

\bibitem{frugal4}
D.~Kempe, M.~Salek, and C.~Moore.
\newblock Frugal and truthful auctions for vertex covers, flows and cuts.
\newblock In {\em FOCS}, pages 745--754, 2010.

\bibitem{LehmannLehmannNisan}
B.~Lehmann, D.~J. Lehmann, and N.~Nisan.
\newblock Combinatorial auctions with decreasing marginal utilities.
\newblock {\em Games and Economic Behavior}, 55(2):270--296, 2006.

\bibitem{dining}
R.~P. Leme, V.~Syrgkanis, and \'{E}va Tardos.
\newblock The dining bidder problem: a la russe et a la francaise.
\newblock {\em SIGecom Exchanges}, 11(2), 2012.

\bibitem{Murota96}
K.~Murota.
\newblock Convexity and steinitz's exchange property.
\newblock {\em Advances in Mathematics}, 124(2):272 -- 311, 1996.

\bibitem{nash_bargaining}
J.~Nash.
\newblock The bargaining problem.
\newblock {\em Econometrica}, 18(2):155--162, April 1950.

\bibitem{PLST12}
R.~{Paes Leme}, V.~Syrgkanis, and {\'E}.~Tardos.
\newblock Sequential auctions and externalities.
\newblock In {\em SODA}, pages 869--886, 2012.

\bibitem{simon_zame}
L.~K. Simon and W.~R. Zame.
\newblock Discontinuous games and endogenous sharing rules.
\newblock {\em Econometrica}, 58(4):861--72, July 1990.

\end{thebibliography}

\columnsversion{}{
\appendix
\section{Proof of Existence of $\epsilon$-Nash}

In this appendix, we prove Theorem \ref{lemma:any_map}.
First, consider the following Lemma:

\begin{lemma}\label{lemma:pre_any_map}
 Given a combinatorial valuation $v$ and a maximal decision map $X$, if $X(p) =
S$, and $p^\epsilon_i = [p_i - \frac{\epsilon}{n}]^+$ for $i \in S$ and
$p^\epsilon_i = p_i$ otherwise, then for any decision map $X'$, and $S^\epsilon
= X'(p^\epsilon)$, it holds that $S \cap \{j; p_j > 0 \} \subseteq
S^\epsilon \subseteq S$.
\end{lemma}

\begin{proof}
First we show that if $S^\epsilon \setminus S \neq \emptyset$,
then $S^\epsilon \notin D(v;p^\epsilon)$, contradicting the definition of
$S^\epsilon$. Since $X$ is maximal, we must have $p_j > 0$ for all $j \in
S^\epsilon \setminus S$.
Thus, if  $S^\epsilon \setminus S \neq \emptyset$, then for $j\in S^\epsilon
\setminus S$, by the second condition in Lemma \ref{lemma:characterization},
$v(S) - p(S) \geq v(S^\epsilon) - p(S^\epsilon) + p_j> v(S^\epsilon) -
p(S^\epsilon)$ so it must also be the case that $v(S) - p^\epsilon(S) >
v(S^\epsilon) - p^\epsilon(S^\epsilon)$ since for $i \in S^\epsilon \setminus
S$, $p^\epsilon_i = p_i$. We conclude that if $S^\epsilon \setminus S \neq
\emptyset$, then $S^\epsilon \notin D(v;p^\epsilon)$. So, it must be the case
that $S^\epsilon \subseteq S$.

Now, we show that $S \cap \{j; p_j > 0\}\subseteq S^\epsilon$. If not, we show
that $S^\epsilon \notin D(v;p^\epsilon)$. Indeed, since $S \in D(v;p)$, then
$v(S) - p(S) \geq v(S^\epsilon) - p(S^\epsilon)$. Now, given that $p^\epsilon
\leq p$ and there is some $j \in S \setminus S^\epsilon$ with $p_j^\epsilon <
p_j$, $v(S) - p^\epsilon(S) > v(S^\epsilon) - p^\epsilon(S^\epsilon)$. 
\end{proof}

\begin{proofof}{Theorem \ref{lemma:any_map}}
Let $S = X(p)$ and let $p^\epsilon$ and $S^\epsilon$ by as in Lemma
\ref{lemma:pre_any_map}. We want to show that $p^\epsilon$ is an
$\epsilon$-Nash equilibrium.

First we consider $i \in S$. Since $S \cap \{j; p_j > 0\} \subseteq S^\epsilon$ (by Lemma \ref{lemma:pre_any_map}), we know that $u^{X'}_i(p^\epsilon) = p_i^\epsilon$. In order to increase his utility by more then $\epsilon$, it is necessary for him to deviate to $\tilde{p}_i > p_i^\epsilon + \epsilon$. Let $\tilde{p} = (\tilde{p}_i, p_{-i}^\epsilon)$. We argue that $i \notin \tilde{S} := X'(\tilde{p})$. Assume otherwise. We note that by the first condition in Lemma \ref{lemma:characterization}, there is $T \not\ni i$, such that $v(T) - p(T) = v(S) - p(S) \geq v(\tilde{S}) - p(\tilde{S})$. We note that: $$\begin{aligned} v(T) - \tilde{p}(T) & = v(T) - p^\epsilon(T) \geq v(T) - p(T)  \\ & \geq v(\tilde{S}) - p(\tilde{S}) > v(\tilde{S}) - \tilde{p}(\tilde{S})  \end{aligned}$$ where the last inequality follows from $\tilde{p}(\tilde{S}) = \tilde{p}_i + p^\epsilon(S' \setminus i) > (p_i + \epsilon - \frac{\epsilon}{n}) + (p(\tilde{S} \setminus i) - (n-1) \frac{\epsilon}{n}) = p(\tilde{S})$. Which contradicts
that $\tilde{S} \in D(v;\tilde{p})$, so it must be the case that $i \notin \tilde{S}$.

Now, consider a seller $i \notin S$. His utility $u^{X'}_i(p^\epsilon) = 0$ since $X'(p^\epsilon) \subseteq S$. In order to increase it by more then $\epsilon$, he needs to deviate to $\tilde{p}_i > \epsilon$. Again, we argue that for $\tilde{p} = (\tilde{p}_i, p_{-i}^\epsilon)$, $i \notin \tilde{S} := X'(\tilde{p})$. Assume otherwise. Since $i \in \tilde{S}$, by the second condition of Lemma \ref{lemma:characterization} we have $v(S) - p(S) \geq v(\tilde{S})-p(\tilde{S} \setminus i)$. Therefore: $$\begin{aligned} v(S) - p^\epsilon(S) & = v(S) - p(S) + [p(S) - p^\epsilon(S)] \\ & \geq v(\tilde{S}) - p(\tilde{S}\setminus i) + [p(\tilde{S} \setminus i) - p^\epsilon(\tilde{S} \setminus i)] \\ &  = v(\tilde{S})-p^\epsilon(\tilde{S} \setminus i) > v(\tilde{S})-\tilde{p}(\tilde{S})\end{aligned}$$
which contradicts the fact that $\tilde{S} \in D(v;\tilde{p})$.\\

In order to see that $S^\epsilon = X'(p^\epsilon)$ and $S = X(p)$ produce the
same welfare, notice that $S^\epsilon \subseteq S$ and all $j \in S \setminus
S^\epsilon$ are priced at zero. Therefore: $v(S \vert S^\epsilon) = 0$.
\end{proofof}

}

\end{document}